\documentclass[11pt,a4paper]{article}%
\usepackage{amsfonts}
\usepackage{amsmath}
\usepackage{amssymb}
\usepackage{graphicx}%
\providecommand{\U}[1]{\protect\rule{.1in}{.1in}}
\newtheorem{theorem}{Theorem}

\newtheorem{corollary}[theorem]{Corollary}

\newtheorem{definition}[theorem]{Definition}
\newtheorem{example}[theorem]{Example}

\newtheorem{proposition}[theorem]{Proposition}
\newtheorem{remark}[theorem]{Remark}

\newenvironment{proof}[1][Proof]{\noindent\textbf{#1.} }{\ \rule{0.5em}{0.5em}}
\begin{document}

\title{On a recent conjecture by Z. Van Herstraeten and N.J. Cerf for the quantum Wigner entropy}
\author{Nuno Costa Dias\textbf{\thanks{ncdias@meo.pt}}
\and Jo\~{a}o Nuno Prata\textbf{\thanks{joao.prata@mail.telepac.pt }}}
\maketitle

\begin{abstract}
We address a recent conjecture stated by Z. Van Herstraeten and N.J. Cerf. They claim that the Shannon entropy for positive Wigner functions is bounded below by a positive constant, which can be attained only by Gaussian pure states. We introduce an alternative definition of entropy for all absolutely integrable Wigner functions, which is the Shannon entropy for positive Wigner functions. Moreover, we are able to prove, in arbitrary dimension, that this entropy is indeed bounded below by a positive constant, which is not very distant from the constant suggested by Van Herstraeten and Cerf. We also prove an analogous result for another conjecture stated by the same authors for the R\'enyi entropy of positive Wigner functions. As a by-product we prove a new inequality for the radar-ambiguity function (and for the Wigner distribution) which is reminiscent of Lieb's inequalities.   

\end{abstract}

\section{Introduction}

In \cite{Cerf1} the authors introduced the concept of Wigner quantum entropy. This entropy amounts to Shannon's entropy associated with a Wigner distribution. Recall the Wigner distribution of some $f \in L^2 (\mathbb{R}^d)$ is given by:
\begin{equation}
Wf (x,p)= \frac{1}{(2 \pi\hbar)^d} \int_{\mathbb{R}^d} f(x+y/2) \overline{f (x-y/2)} e^{- \frac{i}{\hbar}y \cdot p} dy~.
\label{eqIntro1}
\end{equation}
Here $x,p \in \mathbb{R}^d$ denote the position and the momentum of a particle or the conjugate quadrature components of $d$ modes of the electromagnetic field, $h=2 \pi \hbar$ is Planck's constant, and $\overline{c}$ denotes the complex conjugation of $c \in \mathbb{C}$. 

Consider a density matrix (a mixed state) $\widehat{\rho}$, that is a positive trace class operator with unit trace, acting on $L^2(\mathbb{R}^d)$. By the spectral theorem, it can always be written as \cite{Teofanov,Gosson1}:
\begin{equation}
\widehat{\rho} = \sum_j p_j |f_j \rangle \langle f_j|~,
\label{eqIntro2}
\end{equation}
with $p_j \geq 0$, $\sum_j p_j=1$ and states $f_j(x)= \langle x|f_j \rangle \in L^2 (\mathbb{R}^d)$, such that $\langle f_j|f_k \rangle = \delta_{j,k}$. The associated Wigner distribution is:
\begin{equation}
W \rho (x,p)= \frac{1}{(2 \pi\hbar)^d} \sum_j p_j \int_{\mathbb{R}^d} f_j (x+y/2) \overline{f_j (x-y/2)} e^{- \frac{i}{\hbar}y \cdot p} dy~.
\label{eqIntro3}
\end{equation} 
In this fashion we obtain a normalized function:
\begin{equation}
\int_{\mathbb{R}^d} \int_{\mathbb{R}^d} W \rho (x,p) ~dx dp =1~.
\label{eqIntro3A}
\end{equation} 
The Shannon entropy of $W \rho$ is defined by \cite{Cerf1}:
\begin{equation}
H \left[W \rho \right] = - \int_{\mathbb{R}^d}\int_{\mathbb{R}^d} W \rho(x,p) \ln \left( W \rho (x,p)\right) dx dp ~.
\label{eqIntro4}
\end{equation} 
However, one faces an immediate problem with this definition. Although Wigner functions share many properties with true probability densities, they fail to qualify as full-fledged probability densities because of their lack of positivity. For instance, the Wigner function of a pure state (\ref{eqIntro1}) is everywhere non-negative if and only if the state is a generalized Gaussian. This is known as Hudson's theorem \cite{Hudson}. For mixed states the situation is much more intricate \cite{Werner,Mandilara,Narcowich} and there is no known general prescription for obtaining the entire set of positive Wigner distributions. The authors of \cite{Cerf1} therefore chose to restrict the definition of \textit{quantum Wigner entropy} (eq.(\ref{eqIntro4})) to the set of positive Wigner distributions only. They called this set the \textit{Wigner-positive states} \cite{Cerf1,Cerf2}.

They then went on and argued that it is nevertheless a proper measure of quantum uncertainty in phase space and they advocate that it can be an important physical quantity in the context of quantum optics, because it is invariant under affine symplectic transformations (displacements, rotations and squeezing), and it permits one to establish a Wigner entropy-power inequality. It bears information about the joint uncertainty of the marginal distributions of Wigner-positive states and, unlike the Wehrl entropy, it is not the classical entropy of a heterodyne detection.

Moreover they conjecture that the functional (\ref{eqIntro4}) is bounded below within the convex set of Wigner-positive states and that, in $d=1$, the lower bound is $1 +\ln(\pi \hbar)=\ln(e \pi \hbar)$, which is attained by pure state Gaussians. The authors proved the conjecture to be true for the subset of Gaussian Wigner functions and for the subset of the passive phase-invariant states (that is, states from which no work can be extracted through unitary operations \cite{Pusz}, and which have "radial" Wigner functions) of the harmonic oscillator. The latter case is only adapted to dimension $d=1$.

If the conjecture is true for all Wigner-positive states, then it implies both the Hirschman, Beckner, Bialynicki-Birula, Mycielski uncertainty principle \cite{Beckner,Birula,Hirschman} for the marginal distributions, as well as the Wehrl entropy bound \cite{Lieb2,Wehrl}, but it is a stronger condition.

In Appendix A of \cite{Cerf1} the authors also consider the R\'enyi entropy of Wigner distributions,
\begin{equation}
h_{\alpha} \left[W \rho \right] = \frac{\alpha}{1-\alpha} \ln \left(\|W \rho \|_{L^{\alpha} (\mathbb{R}^{2d})}\right)~,
\label{eqExtra1}
\end{equation}
for $1 < \alpha \leq \infty$. They also conjectured that, for positive Wigner functions, this entropy is bounded below and that the minimum can only be attained by Gaussian pure states.

The purpose of this paper is to address these conjectures. We were not able to prove the conjectures, but we were nevertheless able to prove some partial results in that direction, and extend them to a larger class of Wigner functions:

\begin{itemize}
\item We introduce a new functional acting on $W \rho$, denoted by $S\left[W \rho\right]$, which amounts to the Shannon entropy of $|W \rho|$ (up to a normalization constant). This functional is well defined for all Wigner distributions in $L^1 (\mathbb{R}^{2d})$, and, if $W \rho$ is a Wigner-positive state, it coincides with the quantum Wigner entropy (\ref{eqIntro4}) of Van Herstraeten and N.J. Cerf (see Definition \ref{DefinitionWignerEntropy} below). In the same vein we define the Wigner-R\'enyi entropy for arbitrary absolutely integrable Wigner functions (Definition \ref{DefinitionRenyi}).

\item We prove (in arbitrary dimension $d$) that $S\left[W \rho\right] > d \ln (2 \pi \hbar)$. We will argue that the constant cannot be attained by any Wigner function in $L^1 (\mathbb{R}^{2d})$. On the other hand, it is not very different from the conjectured constant (see Theorem \ref{TheoremEntropy}).

\item In Appendix A of \cite{Cerf1} it was conjectured that for the R\'enyi entropy of positive Wigner functions in $d=1$, we have:
\begin{equation}
h_{\alpha } \left[W \rho \right] \geq \ln (\pi \hbar)+\frac{\alpha}{\alpha -1}~,
\label{eqExtra2}
\end{equation}
and that the lower bound can only be attained by Gaussian pure states.

We prove (again in arbitrary dimension $d$) that their conjecture is true if $2 \leq \alpha \leq \infty$. For $1 < \alpha <2$, we prove that the R\'enyi entropy is bounded below but we do not obtain the optimal constant (see Theorem \ref{TheoremRenyi}). We should emphasize that all the results obtained for the Wigner-Renyi entropy $h_{\alpha}$ are valid for all Wigner functions belonging to $L^1 (\mathbb{R}^{2d})$, regardless of being positive or not. 

\item In order to prove the previous results, we derive a new inequality for the radar-ambiguity function\footnote{For the definition of $A(f,g)$ see section 2.} $A(f,g)$ (which is also valid \textit{mutatis mutandis} for the Wigner distribution $W(f,g)$). Lieb \cite{Lieb1} proved a certain inequality for the radar-ambiguity function which holds for $L^q$ norms with $q > 2$ and certain values of $p$ (to be specified later):
\begin{equation}
\|A(f,g)\|_{L^q(\mathbb{R}^{2d})} \leq C \|f\|_{L^p(\mathbb{R}^d)}\|g\|_{L^{p^{\prime}}(\mathbb{R}^d)}~,
\label{LiebIntro1}
\end{equation}
and an opposite inequality for $1 \leq q <2$:
\begin{equation}
\|A(f,g)\|_{L^q(\mathbb{R}^{2d})} \geq C \|f\|_{L^p(\mathbb{R}^d)}\|g\|_{L^{p^{\prime}}(\mathbb{R}^d)}~.
\label{LiebIntro2}
\end{equation}
Our inequality is valid for $1 < q < 2$, but the inequality is opposite to that in (\ref{LiebIntro2}) and requires the appurtenance of $A(f,g)$ to $L^1 (\mathbb{R}^{2d})$ (see Proposition \ref{PropositionNewIneq}):
\begin{equation}
\|A(f,g)\|_{L^q(\mathbb{R}^{2d})} \leq C \|A(f,g)\|_{L^1(\mathbb{R}^{2d})}^{\theta /q} \left(\|f\|_{L^p(\mathbb{R}^d)}\|g\|_{L^{p^{\prime}}(\mathbb{R}^d)} \right)^{1- \theta /q}~,
\label{LiebIntro3}
\end{equation}
where $\theta$ is a certain parameter to be specified later.

Here and henceforth, $p^{\prime}$ denotes the H\"older conjugate index of $p \in \left[1, \infty \right]$:
\begin{equation}
\frac{1}{p} + \frac{1}{p^{\prime}}=1~.
\label{eqPreliminary3B}
\end{equation}
The positive constants $C$ in (\ref{LiebIntro1},\ref{LiebIntro2},\ref{LiebIntro3}) can only depend on $q,p,d,\theta$.

\end{itemize}

\section*{Notation}

Let $\mathcal{S}^{\prime} (\mathbb{R}^d)$ be the space of tempered distributions, which is the dual of the space of Schwartz space of test functions $\mathcal{S} (\mathbb{R}^d)$. Given some open set $\Omega \subset \mathbb{R}^d$, we denote  by $\mathcal{D} (\Omega)$ the space of $C^{\infty}$ test functions with compact support contained in $\Omega$. For $q \in \left[1, \infty \right]$, we denote by $L^q (\mathbb{R}^d)$ the space of all $f \in \mathcal{S}^{\prime} (\mathbb{R}^d)$ for which $\|f \|_{L^q (\mathbb{R}^d)} < \infty$. If $1 \leq q < \infty$, the norm $\|\cdot \|_{L^q (\mathbb{R}^d)}$ is given by:
\begin{equation}
\|f\|_{L^q (\mathbb{R}^d)}= \left(\int_{\mathbb{R}^d} |f(x)| ^q dx \right)^{\frac{1}{q}}~.
\label{eqNotation1}
\end{equation}
Alternatively, if $q=\infty$, then:
\begin{equation}
\|f\|_{L^{\infty} (\mathbb{R}^d)}= \text{ess sup}_{x \in \mathbb{R}^d} |f(x)|~.
\label{eqNotation2}
\end{equation}
If the dimension $d$ is clear from the context, then we will simply write $\|\cdot\|_q$, instead of $\|\cdot\|_{L^q (\mathbb{R}^d)}$.

The same is valid for the inner product, which we may write $\langle \cdot | \cdot \rangle_{L^2 (\mathbb{R}^d)}$ or simply $\langle \cdot | \cdot \rangle$. Let us remark that we shall use the physicists' convention, where the inner product is linear in the second argument and anti-linear in the first: 
\begin{equation}
\langle f | g \rangle:= \int_{\mathbb{R}^d} \overline{f(x)} g(x) dx~.
\label{eqNotation3}
\end{equation}

We denote by $\mathcal{F}_{\hbar}$ the Fourier transform 
\begin{equation}
\left(\mathcal{F}_{\hbar} f \right)(p):= \frac{1}{(2 \pi \hbar)^{d/2}} \int_{\mathbb{R}^d} f(x) e^{- \frac{i}{\hbar} x \cdot p} dx ~,
\label{eqNotation4}
\end{equation}
for $f \in \mathcal{S} (\mathbb{R}^d)$. It extends by usual density arguments to $L^2 (\mathbb{R}^d)$.

A symplectic matrix $S \in Sp(2d; \mathbb{R})$ is a $2d \times 2d$ matrix, such that:
\begin{equation}
SJ S^T =J ~,
\label{eqNotation5}
\end{equation}
where $S^T$ denotes the transpose of $S$ and $J$ is the standard symplectic matrix:
\begin{equation}
 J= \left(
 \begin{array}{c c}
 0_{d \times d} & I_{d \times d}\\
 -I_{d \times d} & 0_{d \times d}
 \end{array}
 \right)~.
\label{eqNotation6}
\end{equation}
Here $0_{d \times d}$ and $I_{d \times d}$ are the zero and the identity $d \times d$ matrices, respectively.

\section{Some preliminary results}

In this section we will introduce several concepts and review some results which will be useful in the sequel. Recall that the cross-Wigner function is defined by
\begin{equation}
W(f,g) (x,p)= \frac{1}{(2 \pi\hbar)^d} \int_{\mathbb{R}^d} f(x+y/2) \overline{g (x-y/2)} e^{- \frac{i}{\hbar}y \cdot p} dy~,
\label{eqPreliminary1}
\end{equation}
for $f,g \in L^2 (\mathbb{R}^d)$. Thus the Wigner function (\ref{eqIntro1}) of a pure state is $Wf(x,p)=W(f,f)(x,p)$.  On the other hand, the ambiguity function is:
\begin{equation}
A(f,g) (\tau, \omega) =  \int_{\mathbb{R}^d}f(t- \tau/2) \overline{g(t+ \tau/2)} e^{-2 \pi i \omega \cdot t} dt~.
\label{eqPreliminary2}
\end{equation}
The Cauchy-Schwarz inequality guarantees that $A$ is uniformly continuous and bounded:
\begin{equation}
\left|A(f,g) (\tau, \omega)\right| \leq \|f\|_{L^2(\mathbb{R}^d)} \|g\|_{L^2(\mathbb{R}^d)}~,
\label{eqPreliminary3}
\end{equation}
for all $\tau, \omega \in \mathbb{R}^d$.

This is a particular case ($p=2$) of the inequality obtained from H\"older's inequality:
\begin{equation}
\left|A(f,g) (\tau, \omega)\right| \leq \|f\|_{L^p(\mathbb{R}^d)} \|g\|_{L^{p^{\prime}}(\mathbb{R}^d)}~,
\label{eqPreliminary3A}
\end{equation}
for $1 \leq p \leq \infty$.

Furthermore, if $f, g \in L^2 (\mathbb{R}^d)$, then $A(f,g) \in L^2 (\mathbb{R}^{2d})$ and the following orthogonality identities hold:
\begin{equation}
\langle A(f_1,g_1)|A(f_2,g_2) \rangle_{L^2 (\mathbb{R}^{2d})} = \langle f_1|f_2 \rangle_{L^2 (\mathbb{R}^{d})} ~ \langle g_2| g_1 \rangle_{L^2 (\mathbb{R}^{d})} ~.
\label{eqPreliminary4}
\end{equation}
Similarly, $W(f,g) \in L^2 (\mathbb{R}^{2d})$ and the Moyal identity holds (see e.g. \cite{Gosson2}):
\begin{equation}
\langle W(f_1,g_1)|W(f_2,g_2) \rangle_{L^2 (\mathbb{R}^{2d})} = \frac{1}{(2 \pi \hbar)^d} \langle f_1|f_2 \rangle_{L^2 (\mathbb{R}^{d})} ~ \langle g_2| g_1 \rangle_{L^2 (\mathbb{R}^{d})} ~.
\label{eqPreliminary4A}
\end{equation}
The cross-Wigner function (\ref{eqPreliminary1}) and the ambiguity function (\ref{eqPreliminary2}) are related by:
\begin{equation}
W(f,g) (x,p)=\frac{1}{(\pi \hbar)^d} A(f,g^-) \left(-2 x, \frac{p}{\pi \hbar} \right)~,
\label{eqPreliminary5}
\end{equation}
where $g^-(x)=g(-x)$.

Notice that, from (\ref{eqIntro3},\ref{eqPreliminary3},\ref{eqPreliminary5}), we have:
\begin{equation}
\|W \rho \|_{\infty} \leq \frac{1}{(\pi \hbar)^d}~.
\label{eq14A}
\end{equation}
Let us also recall the Bastiaans-Littlejohn Theorem \cite{Bastiaans,Littlejohn} concerning Gaussian states. Consider the normalized Gaussian
\begin{equation}
\frac{\sqrt{\det M}}{(\pi \hbar)^d} e^{- \frac{1}{\hbar} (z-z_0) \cdot M(z-z_0)}~.
\label{eq14B}
\end{equation}
where $z=(x,p) \in \mathbb{R}^{2d}$, $z_0 \in \mathbb{R}^{2d}$ is some fixed constant and $M$ is a real, symmetric, positive-definite $2d \times 2d $ matrix. Then (\ref{eq14B}) is the Wigner function associated with some pure state $f \in L^2 (\mathbb{R}^{d})$ if and only if $M$ is a symplectic matrix: $M \in Sp(2d; \mathbb{R})$ .

In \cite{Lieb1} Lieb defined the functional 
\begin{equation}
I_{f,g} (q)= \int_{\mathbb{R}^d}  \int_{\mathbb{R}^d} \left|A(f,g) (\tau, \omega)\right|^q d \tau  d \omega = \|A\|_{L^q (\mathbb{R}^{2d})}^q~,
\label{eqPreliminary10}
\end{equation}
for $q \in \left[2, \infty\right]$.
 
He then proved the inequalities (\ref{LiebIntro1}) and (\ref{LiebIntro2}) mentioned in the Introduction. Let us recapitulate these results in more detail.

In the sequel, the constants $H(p,q,d)$ are given by:
\begin{equation}
H(p,q,d)= 
\left[C_{q^{\prime}} \left(\frac{C_{p/q^{\prime}} C_{p^{\prime}/q^{\prime}}}{C_{q/q^{\prime}}} \right)^{\frac{1}{q^{\prime}}} \right]^d~,
\label{eqConstants1}
\end{equation}
where $C_r$ is the Babenko-Beckner constant \cite{Babenko,Beckner}: 
\begin{equation}
C_r= \sqrt{\frac{r^{1/r}}{(r^{\prime})^{1/r^{\prime}}}}~,
\label{eqConstants2}
\end{equation}
if $r \neq 1$ or $r \neq \infty$ and
\begin{equation}
C_1=C_{\infty}=1~.
\label{eqConstants3}
\end{equation}
Explicitly, we have;
\begin{equation}
H(p,q,d)= \left[ \frac{p p^{\prime}}{q^2} |q-2|^{2-q} ~|q-p|^{-1+ q/p} ~|q-p^{\prime}|^{-1+q/p^{\prime}} \right]^{\frac{d}{2q}}~.
\label{eqConstants4}
\end{equation}

\begin{theorem}\label{TheoremLieb1}
Let $q \geq 2$ and $q^{\prime} \leq p, p^{\prime} \leq q$. If $f \in L^p (\mathbb{R}^d)$ and $g \in L^{p^{\prime}} (\mathbb{R}^d)$, then $A(f,g) \in L^q (\mathbb{R}^{2d})$ and we have:
\begin{equation}
\|A(f,g)\|_{L^q (\mathbb{R}^{2d})} \leq H(p,q,d)\|f\|_{L^p (\mathbb{R}^d)} \|g\|_{L^{p^{\prime}} (\mathbb{R}^d)} ~.
\label{eqLieb11}
\end{equation}
For $q=2$, the previous inequality becomes an identity for all $f,g \in L^2 (\mathbb{R}^d)$. 

If $q>2$, we have an equality if and only if $f,g$ are a matched Gaussian pair:
\begin{equation}
\begin{array}{l}
f(x)=\exp \left(- x \cdot A x + b \cdot x + \gamma \right)\\
\\
g(x)=\exp \left(- x \cdot A x + c \cdot x + \eta \right)~, 
\end{array}
\label{eqLieb12}
\end{equation}
where $A$ is a symmetric $d \times d$ matrix, with $Re(A)>0$, $b,c \in \mathbb{C}^d$ and $\gamma, \eta \in \mathbb{C}$.

\end{theorem}

\begin{remark}\label{Remark1}
This inequality was proved with resort to the H\"older and the sharp Young convolution inequality \cite{Lieb4,Lieb5}. We shall be mainly interested in the $p=p^{\prime}=2$ case, which reads:
\begin{equation}
\|A(f,g)\|_{L^q (\mathbb{R}^{2d})} \leq \left(\frac{2}{q}\right)^{d/q} \|f\|_{L^2 (\mathbb{R}^d)} \|g\|_{L^{2} (\mathbb{R}^d)} ~.
\label{eqLieb12A}
\end{equation}
The case $q=2$, is just the identity:
\begin{equation}
\|A(f,g)\|_{L^2 (\mathbb{R}^{2d})} = \|f\|_{L^2 (\mathbb{R}^d)} \|g\|_{L^{2} (\mathbb{R}^d)} ~,
\label{eqLieb11A}
\end{equation}
for all $f,g \in L^2 (\mathbb{R}^d)$,which is an immediate consequence of (\ref{eqPreliminary4}).
\end{remark}

For $1 \leq q < 2$, Lieb proved an opposite inequality. The proof required Leindler's inequality \cite{Barthe,Lieb4,Leindler}. 

\begin{theorem}\label{TheoremLieb2}
Assume that $t \mapsto f(t-\tau/2) \overline{g(t+\tau/2)} \in L^1 (\mathbb{R}^d)$ for almost all fixed $\tau \in \mathbb{R}^d$, and let $1 \leq q <2$. Suppose that $0 < \|A(f,g)\|_{L^q (\mathbb{R}^{2d})} < \infty$. Then $f \in L^p (\mathbb{R}^d)$ and $g \in L^{p^{\prime}} (\mathbb{R}^d)$ for $q \leq p,p^{\prime} \leq q^{\prime}$, and we have:
\begin{equation}
\|A(f,g)\|_{L^q (\mathbb{R}^{2d})} \geq H(p,q,d)\|f\|_{L^p (\mathbb{R}^d)} \|g\|_{L^{p^{\prime}} (\mathbb{R}^d)} ~.
\label{eqLieb12B}
\end{equation}
Again, the inequality is sharp and can be saturated if and only if $f$ and $g$ are a matched pair of Gaussians.

\end{theorem}

A straightforward corollary of Theorem \ref{TheoremLieb1} is:
\begin{corollary}\label{CorollaryLieb1}
Let $W \rho(x,p)= \sum_j p_j W f_j(x,p) $ be the Wigner function associated with some mixed state (cf.(\ref{eqIntro3})) and let $q \geq 2$. Then the following inequality holds:
\begin{equation}
\|W \rho \|_{L^q (\mathbb{R}^{2d})} \leq \left(\frac{1}{q (\pi \hbar)^{q-1}}\right)^{d/q}~. 
\label{eqLieb13}
\end{equation}
If $q=2$ we have an identity, if and only if $\rho$ is a pure state. If $q>2$, then an equality holds if and only if $W \rho$ is a Gaussian:
\begin{equation}
W \rho (z) = \frac{1}{(\pi \hbar)^d} \exp \left(-\frac{1}{\hbar}(z-z_0) \cdot M (z-z_0) \right)~,
\label{eqLieb14}
\end{equation}
where $z_0 \in \mathbb{R}^{2d}$ and $M$ is a real, symmetric, positive-definite matrix, which is symplectic: $M \in Sp(2d;\mathbb{R})$.
\end{corollary}

\begin{proof}
Assuming $\|f\|_{L^2(\mathbb{R}^d)}=\|g\|_{L^2(\mathbb{R}^d)}=1$, we have from (\ref{eqPreliminary5},\ref{eqLieb12A}):
\begin{equation}
\begin{array}{c}
\| W(f,g) \|_{L^q (\mathbb{R}^{2d})}^q= \int_{\mathbb{R}^{2d}} |W(f,g)(x,p)|^q dx d p =\\
\\
= \left(\frac{1}{\pi \hbar} \right)^{qd} \left(\frac{\pi \hbar}{2} \right)^d \int_{\mathbb{R}^{2d}} |A (f,g^-) (\tau, \omega) |^q d \tau d \omega \leq \left(\frac{1}{q(\pi \hbar)^{q-1}} \right)^d~. 
\end{array}
\label{eqLieb15}
\end{equation}
Thus, by convexity, we have:
\begin{equation}
\|W \rho \|_{L^q (\mathbb{R}^{2d})}= \|\sum_j p_j W f_j \|_{L^q (\mathbb{R}^{2d})} \leq \sum_j p_j \|W f_j \|_{L^q (\mathbb{R}^{2d})} \leq \left(\frac{1}{q(\pi \hbar)^{q-1}} \right)^{d/q}~. 
\label{eqLieb16}
\end{equation}
Moreover, if $q >2$, we have an equality if and only if the state is pure and a Gaussian. By Littlejohn's theorem it follows that $W \rho$ is of the form (\ref{eqLieb14}).

Alternatively, if $q=2$, we have from (\ref{eqPreliminary4A}):
\begin{equation}
\begin{array}{c}
\|W \rho \|_2^2 = \sum_{j,k} p_j p_k\langle W f_j | W f_k\rangle= \frac{1}{(2 \pi \hbar)^d}\sum_{j,k} p_j p_k \delta_{j,k}=\\
\\
= \frac{1}{(2 \pi \hbar)^d} \sum_{j} p_j^2 \leq \frac{1}{(2 \pi \hbar)^d} \sum_{j} p_j =\frac{1}{(2 \pi \hbar)^d}~.
\end{array}
\label{eqLieb16A}
\end{equation}
We have an equality if and only if the state is pure, which concludes the proof. 
\end{proof}
\begin{remark}\label{RemarkPurity}
The quantity
\begin{equation}
\mathcal{P} \left[\rho \right]:= (2 \pi \hbar)^d \|W \rho\|_2^2  \leq 1
\label{eqLieb16B}
\end{equation}
is called the \textit{purity} of the state $\rho$.
\end{remark}

\section{A new inequality for the ambiguity and the Wigner distributions}

We now state and prove our inequality (\ref{LiebIntro3}).
\begin{theorem}\label{TheoremNewInequality1}
Let $q,\theta, p$ be such that $1 < q <2$, $2-q < \theta <1$ and $ \frac{q-\theta}{q-1} \leq p, p^{\prime} \leq \frac{q-\theta}{1- \theta}$. Suppose that $A(f,g) \in L^1 (\mathbb{R}^{2d})$, $f \in  L^p (\mathbb{R}^{d})$ and $g \in  L^{p^{\prime}} (\mathbb{R}^{d})$. Then    
$ A(f,g) \in L^q(\mathbb{R}^{2d})$ and we have:
\begin{equation}
\|A(f,g)\|_{L^q(\mathbb{R}^{2d})} \leq  
\|A(f,g)\|_{L^1(\mathbb{R}^{2d})}^{\theta/q} \left(H \left(p, \frac{q-\theta}{1- \theta},d \right)
\|f\|_{L^p(\mathbb{R}^{d})}
\|g\|_{L^{p^{\prime}}(\mathbb{R}^{d})}\right)^{1- \theta/q}~.
\label{eqNewInequality1}
\end{equation}
\end{theorem}

\begin{proof}
If $f \in L^p(\mathbb{R}^{d})$, $g\in L^{p^{\prime}}(\mathbb{R}^{d})$, then by (\ref{eqPreliminary3A}), $A(f,g) \in L^{\infty} (\mathbb{R}^{2d})$. If, in addition, $A(f,g) \in L^1 (\mathbb{R}^{2d})$, then by familiar interpolation theorems, we also have $A(f,g) \in L^q(\mathbb{R}^{2d})$ for all $q \in \left[1, \infty \right]$.

From H\"older's inequality $\|FG\|_{L^1} \leq \|F\|_{L^r} \|G \|_{L^{r^{\prime}}}$, $\frac{1}{r}+ \frac{1}{r^{\prime}}=1$, with $r= \frac{1}{\theta}$ and $r^{\prime}=\frac{1}{1- \theta}$, we obtain:
\begin{equation}
\begin{array}{c}
\|A(f,g)\|_q^q = \int \int |A(f,g) (\tau, \omega)|^q d \tau d \omega= \\
\\
=\int \int |A(f,g) (\tau, \omega)|^{\theta} |A(f,g) (\tau, \omega)|^{q- \theta} d \tau d \omega \leq \|A(f,g)\|_{1}^{\theta} \|A(f,g)\|_{\frac{q-\theta}{1- \theta}}^{q- \theta}~. 
\end{array}
\label{eqNewInequality2}
\end{equation}
Since, by assumption, $\frac{q-\theta}{1- \theta} >2$, we have from the Lieb inequality (\ref{eqLieb11}) and (\ref{eqNewInequality2}) that:
\begin{equation}
\begin{array}{c}
\|A(f,g)\|_q \leq \|A(f,g)\|_{1}^{\theta/q} \|A(f,g)\|_{\frac{q-\theta}{1- \theta}}^{1- \theta/q}\\
\\
\leq  \|A(f,g)\|_{1}^{\theta/q}  \left(H \left(p, \frac{q-\theta}{1- \theta},d \right) \|f\|_p \|g\|_{p^{\prime}} \right)^{1- \theta/q}~,
\end{array}
\label{eqNewInequality3}
\end{equation}
and the result follows.
\end{proof}

\begin{remark}\label{RemarkNewInequality1}
From the proof, we realize that there are no functions $f,g$ for which we have an equality in (\ref{eqNewInequality1}). Indeed, from H\"older's inequality, we have an identity in (\ref{eqNewInequality2}) if and only if there exists a constant $\lambda \geq 0$, such that \cite{Lieb3}:
\[
\left(|A(f,g)(z)|^{q- \theta}\right)^{\theta}=\lambda \left(|A(f,g) (z)|^{\theta}\right)^{1- \theta} \Leftrightarrow |A (f,g)(z)|^{(q-1) \theta}=\lambda~, 
\] 
for all $z=(\tau , \omega) \in \mathbb{R}^{2d}$. If $A(f,g)$ does not vanish identically, then this is only possible if either $\theta=0$ or $q=1$ (recall that $A(f,g) \in L^2 (\mathbb{R}^{2d})$). However, both values are outside of the range considered in Theorem \ref{TheoremNewInequality1}.

The case $p=p^{\prime}=2$,
\begin{equation}
\|A(f,g)\|_{L^q(\mathbb{R}^{2d})} \leq \left[\frac{2(1- \theta)}{q - \theta} \right]^d 
\|A(f,g)\|_{L^1(\mathbb{R}^{2d})}^{\theta/q}
\|f\|_{L^2(\mathbb{R}^{d})}
\|g\|_{L^2(\mathbb{R}^{d})}~,
\label{eqNewInequalityRemark1}
\end{equation}
will be useful in the sequel.
\end{remark}

A simple consequence of the proof of this theorem is the analogous result for the cross-Wigner distribution $W(f,g)$ and for the Wigner function of a mixed state $W \rho$:
\begin{proposition}\label{PropositionNewIneq}
Let $q,\theta$ be such that $1 < q <2$ and $2-q < \theta <1$. Suppose that $W(f,g) \in L^1 (\mathbb{R}^{2d})$ and $f,g \in L^2 (\mathbb{R}^{d})$. Then    
$ W(f,g) \in L^q(\mathbb{R}^{2d})$ and we have:
\begin{equation}
\begin{array}{c}
\|W(f,g)\|_{L^q(\mathbb{R}^{2d})} \\
\\
\leq  \left(\frac{\pi \hbar}{2} \right)^{d/q} \left[\frac{2(1- \theta)}{\pi \hbar (q- \theta)} \right]^d \left(2^d \|W(f,g)\|_{L^1(\mathbb{R}^{2d})} \right)^{\theta/q} \|f\|_{L^2(\mathbb{R}^{d})} \|g\|_{L^2(\mathbb{R}^{d})} ~.
\end{array}
\label{eqNewIneq1}
\end{equation}

On the other hand, if $W \rho$ is of the form (\ref{eqIntro3}), and $W \rho \in L^1(\mathbb{R}^{2d})$, then $W \rho \in L^q(\mathbb{R}^{2d})$ and:
\begin{equation}
\|W \rho \|_{L^q(\mathbb{R}^{2d})} \\
\\
\leq  \frac{1}{(\pi \hbar)^{d(1-1/q)}} \left(\frac{1- \theta}{q- \theta} \right)^{\frac{d}{q}(1- \theta)} \|W \rho\|_{L^1(\mathbb{R}^{2d})} ~.
\label{eqNewIneq2}
\end{equation}
\end{proposition}

\begin{proof}
From (\ref{eqPreliminary5}) and (\ref{eqNewInequalityRemark1}), we immediately obtain (\ref{eqNewIneq1}).

Recall that $W \rho \in L^2 (\mathbb{R}^{2d})$ (cf.(\ref{eqLieb16B})). If, in addition, $W\rho \in L^1 (\mathbb{R}^{2d})$, then we also have $W \rho \in L^q (\mathbb{R}^{2d})$ for all $1 <q<2$.

Following the same steps as in (\ref{eqNewInequality2}), we have:
\begin{equation}
\begin{array}{c}
\|W \rho\|_q^q \leq \|W \rho\|_1^{\theta} ~\|W \rho \|_{\frac{q-\theta}{1-\theta}}^{q- \theta}\\
\\
\Leftrightarrow \frac{\|W \rho\|_q^q }{\|W \rho\|_1^q} \leq \left(\frac{\|W \rho \|_{\frac{q-\theta}{1-\theta}}}{\|W \rho\|_1}\right)^{q-\theta} \leq \|W \rho \|_{\frac{q-\theta}{1-\theta}}^{q-\theta}~.
\end{array}
\label{eqNewIneq3}
\end{equation}
In the last inequality, we used the fact that (cf.(\ref{eqIntro3A})):
\begin{equation}
\|W \rho \|_1 = \int_{\mathbb{R}^{2d}} |W \rho(z)| dz\geq \left|\int_{\mathbb{R}^{2d}} W \rho(z)dz \right|=1~, 
\label{eqNewIneq4}
\end{equation}
with equality if and only if $W \rho$ is everywhere non-negative.

Since $\frac{q-\theta}{1- \theta} >2$, we obtain from (\ref{eqLieb13}):
\begin{equation}
\frac{\|W \rho\|_{q}^q}{\|W \rho\|_{1}^q}  
\leq \frac{1}{(\pi \hbar)^{d(q-1)} \left(\frac{q- \theta}{1- \theta} \right)^{d(1- \theta)}}~,
\label{eqLieb17E} 
\end{equation}
and the result follows. 
\end{proof}
 
\section{Definition of quantum Wigner entropy and of Wigner-R\'enyi entropy}

For Wigner-positive states, Van Herstraeten and Cerf defined the Wigner quantum entropy to be their Shannon entropy. Here we wish to extend the definition to a larger class of states. 

In the spirit of \cite{Gosson1}, we define the \textit{Feichtinger states} to be the set of density matrices whose Wigner distribution $W \rho$ is absolutely integrable:
\begin{equation}
W \rho \in L^1 (\mathbb{R}^{2d})~.
\label{eqWignerEntro1}
\end{equation} 
This terminology comes from the fact that, for pure states, we have:
\begin{equation}
W f \in L^1 (\mathbb{R}^{2d}) \Leftrightarrow f \in \mathcal{S}_0 (\mathbb{R}^d)~,
\label{eqWignerEntro2}
\end{equation} 
where $\mathcal{S}_0 (\mathbb{R}^d)$ is Feichtinger's algebra \cite{Feichtinger}\footnote{See \cite{Jakobsen} for a recent review on some of the main properties of the Feichtinger algebra.}.  

We shall denote by $\mathcal{F}(\mathbb{R}^{2d})$ the convex set of Wigner functions associated with Feichtinger states. Obviously, all Wigner-positive states belong to $\mathcal{F}(\mathbb{R}^{2d})$ (see (\ref{eqNewIneq4})). 

In the sequel, we define two positive measures. The measure $\mu_{\rho}(z)$ is a probability measure, in terms of which we will state our main result. The measure $\nu_{\rho}(z)$, on the other hand, is not normalized, and it is only an auxiliary measure. The reason for introducing it resides in the fact that $\nu_{\rho}(z) \leq 1$. This condition will be important in the proof of Theorem \ref{TheoremEntropy}. Some properties of $\mu_{\rho}(z)$ and $\nu_{\rho}(z)$ will be stated in Proposition \ref{Proposition1} below. 

\begin{definition}\label{DefinitionMeasures}
Let $W \rho \in \mathcal{F}(\mathbb{R}^{2d})$. We define
\begin{equation}
\mu_{\rho}(z):= \frac{|W \rho(z)|}{\|W \rho \|_{L^1 (\mathbb{R}^{2d})}}~, \hspace{1 cm} \nu_{\rho}(z):=(\pi \hbar)^d \mu_{\rho} (z)= (\pi \hbar)^d \frac{|W \rho(z)|}{\|W \rho \|_{L^1 (\mathbb{R}^{2d})}}~.
\label{eqLieb17}
\end{equation} 
\end{definition}

\begin{definition}\label{DefinitionWignerEntropy}
Let $W \rho \in \mathcal{F}(\mathbb{R}^{2d})$. Its \textit{quantum Wigner entropy} is defined by:
\begin{equation}
\begin{array}{c}
S \left[W \rho\right]:= H \left[\mu_{\rho} \right]= -\int_{\mathbb{R}^{2d}} \mu_{\rho}(z) \ln \left(\mu_{\rho} (z) \right) dz= \\
\\
= -\int_{\mathbb{R}^{2d}} \frac{| W \rho (z)|}{\| W \rho \|_{L^1 (\mathbb{R}^{2d})}} \ln \left(\frac{| W \rho (z)|}{\| W \rho \|_{L^1 (\mathbb{R}^{2d})}}\right) dz~.
\end{array}
\label{eqWignerEntro3}
\end{equation}
\end{definition}
Notice that for a Wigner-positive state $\| W \rho \|_{L^1 (\mathbb{R}^{2d})}=1$, and hence $\mu_{\rho}=W \rho$. Thus, in the set of Wigner-positive states, the definition (\ref{eqWignerEntro3}) coincides with the definition of 
Van Herstraeten and Cerf.

Let us briefly discuss our definition of entropy. For all practical purposes this definition plays only an auxiliary r\^ole in the proof of our main result. It is not clear at all at this point whether it is a pertinent measure of entropy. 

\begin{itemize}
\item First of all notice that the measure $\mu_{\rho}$ does not have the correct (quantum) marginal distributions,
\begin{equation}
\int_{\mathbb{R}^d} \mu_{\rho}(x,p) dp \neq \sum_j p_j |f_j (x)|^2 ~, \hspace{1 cm} \int_{\mathbb{R}^d} \mu_{\rho}(x,p) dx \neq \sum_j p_j |\mathcal{F}_{\hbar} f_j (p)|^2 ~,
\label{eqWignerEntro4}
\end{equation}
unless $W \rho$ is a Wigner-positive state. This fact is illustrated in Example \ref{Example1} of the Appendix for the first excited state of the harmonic oscillator. It should however be emphasized that neither does the Husimi function appearing in the Wehrl entropy. 

\item Our measure of entropy is (in general) not concave for convex combinations of Wigner functions. The exception are of course convex combinations of Wigner-positive states.

It is not easy to construct explicit examples that illustrate this. The Shannon entropy is difficult to compute even for the Hermite functions (other than the ground state) \cite{Sanchez}. Without calculating the quantum Wigner entropy for some specific states, we are nevertheless able to prove that our entropy is not concave in general (see Example \ref{ExampleConcavity} in the Appendix). 
\end{itemize}

\vspace{0.3 cm}
On a more positive note, we have:

\begin{itemize}
\item The entropy satisfies the following extensive property. If the density matrix $\widehat{\rho}$ factorizes as the product $\widehat{\rho}=\widehat{\rho}_A \otimes \widehat{\rho}_B$ of two density matrices of subsystems $A$ and $B$, then, since $\mu_{\rho}=\mu_{\rho_A}\mu_{\rho_B}$, we have:
\begin{equation}
S \left[W \rho \right]= S \left[W \rho_A \right]+S \left[W \rho_B \right]~.
\label{eqExtensive}
\end{equation}

\item Let $\widehat{S}$ be a metaplectic operator (Gaussian unitaries) which projects onto the symplectic matrix $S \in Sp(2d; \mathbb{R})$ \cite{Gosson2}. Let $\widehat{\rho}$ be some density matrix and consider the map: $\widehat{\rho} \mapsto \widehat{\rho}^{\prime}=\widehat{S}\widehat{\rho}\widehat{S}^{\ast}$, which is another density matrix. Moreover, let $W \rho$, $W \rho^{\prime}$ be the corresponding Wigner distributions. It is a well known fact that the Weyl quantization satisfies the following symplectic covariance property:
\begin{equation}
W \rho^{\prime}(z)= W \rho (S^{-1}z)~.
\label{eqSymplcticCovariance1}
\end{equation}
From this identity and the fact that $\det S=1$, we conclude that our entropy is the same for $ \widehat{\rho}$ and $ \widehat{\rho}^{\prime}$: 
\begin{equation}
-\int_{\mathbb{R}^{2d}} \mu_{\rho^{\prime}} (z) \ln \left(\mu_{\rho^{\prime}} (z)\right) dz=-\int_{\mathbb{R}^{2d}} \mu_{\rho} (z) \ln \left(\mu_{\rho} (z)\right) dz~.
\label{eqSymplectic2}
\end{equation} 
Consequently, our entropy is symplectically invariant. This is an important property as advocated by the authors of \cite{Cerf1}.

\item Our definition is reminiscent of an entropy introduced by Lieb in \cite{Lieb1}, which shares many of the advantages and shortcomings of our entropy. Lieb considered the following entropy in terms of the radar-ambiguity function:
\begin{equation}
S \left[|A(f,g)|^2 \right]= - \int_{\mathbb{R}^d}\int_{\mathbb{R}^d} |A(f,g)(\tau, \omega)|^2 \ln \left(|A(f,g)(\tau, \omega)|^2\right) d \tau d \omega~,
\label{eqLiebentropy}
\end{equation}
where we recall that $\|A(f,g)\|_2^2 =1$, if $\|f\|_2=\|g\|_2=1$.
\end{itemize}

Finally, let us also define the Wigner-R\'enyi entropy associated with some $W\rho \in \mathcal{F} (\mathbb{R}^{2d})$.

Recall that the R\'enyi entropy for a Wigner function $W \rho$ is given by:
\begin{equation}
h_{\alpha} \left[W \rho \right]:=\frac{\alpha}{1- \alpha} \ln \left(\|W \rho\|_{L^{\alpha}(\mathbb{R}^{2d})} \right)~,
\label{eqAppendixB1}
\end{equation}
for $\alpha \in \left( \right.1, \infty \left. \right]$. 

Notice that this definition holds even if $W \rho$ is not Wigner-positive. 
\begin{definition}\label{DefinitionRenyi}
Let $\alpha \in \left( \right.1, \infty \left. \right]$. For $W\rho \in \mathcal{F} (\mathbb{R}^{2d})$ we define its Wigner-R\'enyi entropy by:
\begin{equation}
H_{\alpha} \left[W \rho \right]:=h_{\alpha} \left[\mu_{\rho} \right]=\frac{\alpha}{1- \alpha} \ln \left(\|\mu_{\rho}\|_{L^{\alpha}(\mathbb{R}^{2d})} \right)~.
\label{eqDefinitionRenyi1}
\end{equation}
\end{definition}

Of course, for Wigner-positive states, we have  $H_{\alpha} \left[W \rho \right]=h_{\alpha} \left[W \rho \right]$.

Let us also remark that:
\begin{equation}
\lim_{\alpha \to 1^+} H_{\alpha} \left[W \rho \right]= S \left[W \rho \right]~.
\label{eqDefinitionRenyi2}
\end{equation}

\section{The main results}

In this section we prove the main results concerning the lower bounds for the quantum Wigner entropy and for the Wigner-R\'enyi entropy. 

\subsection{Quantum Wigner entropy}

The properties of the measures $\mu_{\rho}, \nu_{\rho}$, which are relevant for the sequel, are stated in the following proposition.
\begin{proposition}\label{Proposition1}
The measures $\mu_{\rho}$ and $\nu_{\rho}$ satisfy the following:
\begin{enumerate}
\item $0 \leq \nu_{\rho} (z) \leq 1$, for all $z \in \mathbb{R}^{2d}$.

\item $H \left[\nu_{\rho} \right] = (\pi \hbar)^d \left[H \left[\mu_{\rho} \right] -d \ln (\pi \hbar) \right]$.

\item Let $q \in \left[2, \infty\right]$. Then:
\begin{equation}
\| \nu_{\rho} \|_{L^q (\mathbb{R}^{2d})} \leq \left(\frac{\pi \hbar}{q} \right)^{d/q} ~.
\label{eqLieb17A}
\end{equation} 
Moreover, for $q=2$, we have an equality, if and only if $\rho$ is a pure state and, when $q>2$, if and only if $W \rho$ is of the form (\ref{eqLieb14}).

\item Let $1 \leq q <2$. If $q=1$, then:
\begin{equation}
\| \nu_{\rho} \|_{L^1 (\mathbb{R}^{2d})}=(\pi \hbar)^d~.
\label{eqLieb17AA}
\end{equation} 
Alternatively, if $1 < q < 2$, then:
\begin{equation}
\| \nu_{\rho} \|_{L^q (\mathbb{R}^{2d})} < \left[ \pi \hbar \left(\frac{1- \theta}{q- \theta} \right)^{1- \theta} \right]^{d/q} ~.
\label{eqLieb17B} 
\end{equation} 
for all $2-q < \theta <1$. Moreover, there is no Wigner distribution $W \rho$ for which we have an equality in (\ref{eqLieb17B}).
\end{enumerate}
\end{proposition}

\begin{proof}
\begin{enumerate}
\item Since $\|W \rho \|_{L^{\infty} (\mathbb{R}^{2d})} \leq \frac{1}{(\pi \hbar)^d}$ and $\|W \rho \|_{L^1 (\mathbb{R}^{2d})} \geq 1$, we conclude that
\[
\nu_{\rho} (z) =(\pi \hbar)^d \frac{|W \rho(z)|}{\|W \rho \|_{L^1 (\mathbb{R}^{2d})}} \leq (\pi \hbar)^d \frac{\|W \rho \|_{L^{\infty} (\mathbb{R}^{2d})}}{\|W \rho \|_{L^1 (\mathbb{R}^{2d})}} \leq 1~,
\]
for all $z \in \mathbb{R}^{2d}$.

\item We have:
\[
\begin{array}{c}
H \left[\nu_{\rho} \right] = - \int_{\mathbb{R}^{2d}} \nu_{\rho} (z) \ln \left(\nu_{\rho}(z) \right) = - \int_{\mathbb{R}^{2d}} (\pi \hbar)^d\mu_{\rho} (z) \ln \left((\pi \hbar)^d\mu_{\rho}(z) \right)=\\
\\
= - (\pi \hbar)^d \ln \left((\pi \hbar)^d\right) \int_{\mathbb{R}^{2d}} \mu_{\rho} (z)-(\pi \hbar)^d\int_{\mathbb{R}^{2d}} \mu_{\rho} (z) \ln \left(\mu_{\rho}(z) \right) =\\
\\
= - d(\pi \hbar)^d \ln \left(\pi \hbar\right) + (\pi \hbar)^d H \left[\mu_{\rho} \right]~.
\end{array}
\]

\item If $q \geq 2$, we have from (\ref{eqLieb13},\ref{eqNewIneq4}):
\[
\|\nu_{\rho} \|_q^q = (\pi \hbar)^{qd}\frac{\|W \rho \|_{L^q (\mathbb{R}^{2d})}^q}{\|W \rho \|_{L^1 (\mathbb{R}^{2d})}^q}\leq \left( \frac{\pi \hbar}{q} \right)^d~.
\]
Notice that, for $q>2$, the upper bound $\left( \frac{1}{q (\pi \hbar)^{q-1}} \right)^d$ of the numerator is attained, if and only if $W \rho$ is of the form (\ref{eqLieb14}). Moreover, $\|W \rho \|_{L^1 (\mathbb{R}^{2d})}=1$, if and only if $ W \rho(z)$ is everywhere non-negative, and the result follows. 

If $q=2$, then we have an equality if and only if $\rho$ is a pure state (cf. remark \ref{RemarkPurity}).

\item If $q=1$, the result is obvious. Alternatively, suppose that $1 < q <2$. Let $2-q <\theta<1$.  Then the result follows directly from (\ref{eqNewIneq2}).

The fact that there is no $\rho$ for which we have an equality in (\ref{eqLieb17B}) can be proved following the arguments of Remark \ref {RemarkNewInequality1}. Indeed, an equality in (\ref{eqLieb17B}) would follow from an equality in (\ref{eqNewIneq3}), which is not possible unless $\theta=0$ or $q=1$. 
\end{enumerate}

This completes the proof.
\end{proof}

\vspace{0.3 cm}
\noindent
We are now in a position to state our main result for the quantum Wigner entropy:
\begin{theorem}\label{TheoremEntropy}
We have:
\begin{equation}
S \left[W\rho \right] \geq  d \ln (2\pi \hbar)~,
\label{eqEntropy1}
\end{equation}
for all $W \rho \in \mathcal{F}(\mathbb{R}^{2d})$.
\end{theorem}

\begin{proof}
We shall follow closely the techniques used by Lieb to prove Wehrl's conjecture \cite{Lieb2}. Let 
\begin{equation}
\theta=1 -\epsilon~, ~ q =1+ \epsilon + \epsilon^2~, 
\label{eqEntropy2}
\end{equation}
where $\epsilon >0$ is assumed to be sufficiently small (we will eventually take the limit $\epsilon \to 0$), so that:
\begin{equation}
0< 2-q < \theta <1~.
\label{eqEntropy3}
\end{equation}
Let
\begin{equation}
K_{\epsilon}:= \frac{1}{\epsilon + \epsilon^2} \left(J_{\rho} (1)- J_{\rho} (1+ \epsilon+ \epsilon^2 ) \right) ~, 
\label{eqEntropy4}
\end{equation}
where we denote by $J_{\rho} (q)$ the functional:
\begin{equation}
J_{\rho} (q):= \| \nu_{\rho}\|_q^q~.
\label{eqEntropy4A}
\end{equation}
In view of (\ref{eqLieb17B},\ref{eqEntropy2},\ref{eqEntropy3}), we have:
\begin{equation}
\begin{array}{c}
K_{\epsilon} >  \frac{1}{\epsilon + \epsilon^2} \left[(\pi \hbar)^d- \frac{(\pi \hbar)^d}{(2+ \epsilon)^{\epsilon d}} \right] = \frac{(\pi \hbar)^d}{\epsilon + \epsilon^2}\left[1- \frac{1}{(2+ \epsilon)^{\epsilon d}} \right]\\
\\
= (\pi \hbar)^dd \ln (2) + \mathcal{O} (\epsilon)~,
\end{array}
\label{eqEntropy5}
\end{equation}
where we used the fact that $J_{\rho} (1)= (\pi \hbar)^d$.

Now consider the inequalities \cite{Lieb2}:
\begin{equation}
0 \leq \frac{1}{\sigma} X (1- X^{\sigma}) \leq - X \ln (X)~,
\label{eqEntropy6}
\end{equation}
which hold for $\sigma >0$ and all $X \in \left[0,1 \right]$ (for $X=0$, we set $0 ln(0)=0$).

If we let $\sigma= \epsilon+\epsilon^2$ and $X= \nu_{\rho}(z)$, then:
\begin{equation}
0 \leq \frac{1}{\epsilon+\epsilon^2}\nu_{\rho} (z) \left(1- \left(\nu_{\rho} (z) \right)^{\epsilon+\epsilon^2} \right) \leq - \nu_{\rho}(z) \ln \left(\nu_{\rho} (z) \right)~.
\label{eqEntropy7}
\end{equation}
Upon integration over $z$ on $\mathbb{R}^{2d}$, we obtain:
\begin{equation}
\begin{array}{c}
0 \leq \frac{1}{\epsilon+\epsilon^2} \left[J_{\rho}(1)-J_{\rho}(1+\epsilon+\epsilon^2) \right] \leq  - \int_{\mathbb{R}^{2d}}\nu_{\rho}(z) \ln \left(\nu_{\rho} (z) \right) dz\\
\\
\Leftrightarrow 0 \leq K_{\epsilon} \leq H \left[\nu_{\rho} \right]~,
\end{array}
\label{eqEntropy8}
\end{equation}
which is valid for all $\epsilon>0$. When we take the limit $\epsilon \to 0$ in (\ref{eqEntropy8}) and use (\ref{eqEntropy5}), we obtain:
\begin{equation}
H \left[\nu_{\rho} \right] \geq (\pi \hbar)^d d \ln (2)~.
\label{eqEntropy9}
\end{equation}
Finally, from Proposition \ref{Proposition1}.2, we recover (\ref{eqEntropy1}).
\end{proof}

\subsection{Wigner-R\'enyi entropy}

We now address the conjecture stated in  the Appendix A of \cite{Cerf1} for the Wigner-R\'enyi entropy of Wigner-positive states in dimension $d=1$.

The authors of \cite{Cerf1} conjectured that:
\begin{equation}
h_{\alpha} \left[W \rho \right] \geq \ln (\pi \hbar)+ \frac{\ln (\alpha)}{\alpha-1}~, 
\label{eqAppendixB2}
\end{equation}
and that the minimum can only be attained by Gaussian pure states.

Using our results we can also make a small contribution to this problem. 

\begin{theorem}\label{TheoremRenyi}
Let $W \rho \in \mathcal{F}(\mathbb{R}^{2d})$. 

\begin{itemize}
\item Suppose that $2 \leq \alpha \leq \infty$. Then
\begin{equation}
H_{\alpha} \left[W \rho \right] \geq d \left(\ln (\pi \hbar)+ \frac{\ln (\alpha)}{\alpha-1}\right)~, 
\label{eqAppendixB2A}
\end{equation}
and the minimum can only be attained by Gaussian pure states.

\item For $1< \alpha < 2$, we have:
\begin{equation}
 H_{\alpha} \left[W \rho \right] \geq d \ln (2\pi \hbar) ~.
\label{eqAppendixB7}
\end{equation}
Moreover, there is no Wigner function $W \rho$ for which the lower bound of (\ref{eqAppendixB7}) can be attained.
\end{itemize}
\end{theorem}

\begin{remark}\label{RemarkRenyi}
We note that our bound in eq.(\ref{eqAppendixB7}) is not very distant from the conjectured bound given that:
\begin{equation}
\ln(2)< \frac{\ln (\alpha)}{\alpha-1} <1~,
\label{eqRemarkRenyi}
\end{equation}
for all $1< \alpha <2$.
\end{remark}

\begin{proof}
\begin{itemize}
\item We start with the case $\alpha \geq 2$ and use Corollary \ref{CorollaryLieb1}:
\begin{equation}
 \|\mu_{\rho} \|_{L^{\alpha} (\mathbb{R}^{2d})}= \frac{\|W \rho \|_{L^{\alpha} (\mathbb{R}^{2d})}}{\|W \rho  \|_{L^{1} (\mathbb{R}^{2d})}} \leq \|W \rho \|_{L^{\alpha} (\mathbb{R}^{2d})} \leq \left(\frac{1}{\alpha (\pi \hbar)^{\alpha -1}} \right)^{d/{\alpha}}~.
\label{eqAppendixB4}
\end{equation}
It follows that:
\begin{equation}
H_{\alpha} \left[W \rho \right] \geq d \left(\ln (\pi \hbar)+ \frac{\ln (\alpha)}{\alpha -1} \right)~,
\label{eqAppendixB5}
\end{equation}
which yields (\ref{eqAppendixB2}) for $d=1$.

If $\alpha >2$, then the second inequality in (\ref{eqAppendixB4}) becomes an equality only for Gaussian pure states, in which case the first inequality also becomes an equality. 

If $\alpha =2$, the second inequality in (\ref{eqAppendixB4}) becomes an equality if and only if $W \rho$ is a pure state. The first inequality becomes an equality if and only if $W \rho$ is Wigner-positive. By Hudson's theorem a pure state is Wigner-positive if and only if it is a Gaussian state.

This proves the validity of the Van Herstraeten-Cerf conjecture for the Wigner-R\'enyi entropy in the case $\alpha \in \left[2, \infty\right]$.  

\vspace{0.2 cm}
\item Next consider the case $1 < \alpha <2$. From Proposition \ref{PropositionNewIneq} we have:
\begin{equation}
\|\mu_{\rho} \|_{L^{\alpha} (\mathbb{R}^{2d})}\leq \frac{1}{(\pi \hbar)^{d(1-1/ \alpha)}} \left(\frac{1- \theta}{\alpha - \theta}\right)^{\frac{d}{\alpha}(1-\theta)}~. 
\label{eqAppendixB6}
\end{equation}
It follows that:
\begin{equation}
H_{\alpha} \left[W \rho\right] \geq d \left[\ln(\pi \hbar)+\frac{1-\theta}{1- \alpha}\ln \left(\frac{1-\theta}{\alpha - \theta}\right)\right]~,
\label{eqAppendixC1}
\end{equation}
for all $2- \alpha<\theta<1$.

By inspection, we conclude that the function $\xi : \left(2- \alpha , 1 \right) \to \mathbb{R}$, 
\begin{equation}
\xi(\theta)=\frac{1-\theta}{1- \alpha}\ln \left(\frac{1-\theta}{\alpha - \theta}\right)~,
\label{eqAppendixC1}
\end{equation}
has supremum $\ln(2)$ (at $\theta \to (2-\alpha)^+$), and (\ref{eqAppendixB7}) follows.

This proves that in the case $1 < \alpha <2$ the Wigner-R\'enyi entropy is also bounded below. However, as in the case of the quantum Wigner entropy, we cannot obtain the conjectured bound.
\end{itemize}
\end{proof}

Before we conclude, let us mention that the result for the quantum Wigner entropy can be obtained in the case $\alpha \to 1$. This is of course expected since this limit applied to the R\'enyi entropy yields the Shannon entropy: $\lim_{\alpha \to 1^+}H_{\alpha} \left[W \rho \right]=S \left[W \rho \right]$.

To see this, set as in (\ref{eqEntropy2}):
\begin{equation}
\theta =1 - \epsilon~, ~ \alpha = 1 + \epsilon + \epsilon^2~,
\label{eqAppendixB8}
\end{equation}
for sufficiently small $\epsilon >0$.

We obtain:
\begin{equation}
 H_{1 + \epsilon + \epsilon^2} \left[W \rho \right] \geq d \left(\ln (\pi \hbar) + \frac{1}{1+ \epsilon}\ln \left(2+\epsilon \right) \right)~.
\label{eqAppendixB9}
\end{equation}
If we take the limit $\epsilon \to 0^+$, we obtain:
\begin{equation}
S \left[W \rho \right] \geq d \ln (2 \pi \hbar)~,
\label{eqAppendixB10}
\end{equation}
in accordance with (\ref{eqEntropy1}).

\section*{Acknowledgements}

The authors would like to thank an anonymous referee for useful suggestions and for pointing out to us that Proposition \ref{PropositionNewIneq} would be useful to address the Wigner-R\'enyi conjecture of \cite{Cerf1} and derive the bound for the quantum Wigner entropy in the limit $\alpha \to 1^+$.

\section*{Appendix}

In this appendix we give examples which illustrate (i) the fact that $\mu_{\rho}$ does not yield in general the correct marginals, and (ii) the absence of concavity of the quantum Wigner entropy for states which are not Wigner-positive. 

\begin{example}\label{Example1}
Let 
\begin{equation}
h_1(x)=\frac{1}{\sqrt[4]{\pi \hbar}} \sqrt{\frac{2}{\hbar}}~x e^{- \frac{x^2}{2 \hbar}}
\label{eqExample11}
\end{equation}
denote the Hermite function, which is the first excited state of the harmonic oscillator in $d=1$ (for $m=\omega=1$). The associated Wigner function is given by:
\begin{equation}
W h_1(z)=\frac{1}{\pi \hbar} \left(\frac{2z^2}{\hbar}-1\right) e^{-\frac{z^2}{\hbar}}~.
\label{eqExample11}
\end{equation}
Upon integration of the momentum we obtain:
\begin{equation}
\int_{\mathbb{R}}W h_1(x,p)~ dp= \frac{2x^2}{\hbar \sqrt{\pi \hbar}} e^{- \frac{x^2}{\hbar}} =|h_1(x)|^2~,
\label{eqExample12}
\end{equation}
which is the correct marginal density function for the position variable.

Next we compute $\mu_{h_1}(z)=\frac{|Wh_1(z)|}{\|Wh_1\|_{L^1(\mathbb{R})}}$. 

Using polar coordinates $x=r\cos(\theta)$, $y=r \sin (\theta)$, we obtain:
\begin{equation}
\begin{array}{c}
\|Wh_1\|_{L^1(\mathbb{R})}= \int_{\mathbb{R}}|Wh_1(z)|dz=\\
\\
=\frac{2}{\hbar} \int_0^{\sqrt{\hbar /2}} r\left(1- \frac{2r^2}{\hbar} \right) e^{- \frac{r^2}{\hbar}} dr+\frac{2}{\hbar} \int_{\sqrt{\hbar /2}}^{\infty} r\left(\frac{2r^2}{\hbar}-1 \right) e^{- \frac{r^2}{\hbar}} dr=\\
\\
=\frac{4}{\sqrt{e}}-1~.
\end{array}
\label{eqExample13}
\end{equation}
Consequently:
\begin{equation}
\mu_{h_1}(z)=\frac{|Wh_1(z)|}{\|Wh_1\|_{L^1(\mathbb{R})}}=\frac{e^{-\frac{z^2}{\hbar}}}{\pi \hbar\left(4/\sqrt{e}-1 \right)} \left|\frac{2z^2}{\hbar}-1 \right|~.
\label{eqExample14}
\end{equation}
Let us now evaluate the integral of (\ref{eqExample14}) over $p$. If $|x| \geq \sqrt{\frac{\hbar}{2}}$, then:
\begin{equation}
\begin{array}{c}
\int_{\mathbb{R}} \mu_{h_1} (x,p) dp= \frac{e^{-\frac{x^2}{\hbar}}}{\pi \hbar\left(4/\sqrt{e}-1 \right)}\int_{\mathbb{R}} \left(\frac{2x^2}{\hbar}+\frac{2p^2}{\hbar} -1 \right) e^{-\frac{p^2}{\hbar}}dp =\\
\\
= \frac{x^2 e^{-\frac{x^2}{\hbar}}}{ \hbar \sqrt{\pi \hbar}\left(1/\sqrt{e}-1/4 \right)}~.
\end{array}
\label{eqExample15}
\end{equation}
At this point it is already evident that we do not obtain the correct marginal distribution (\ref{eqExample12}).

For $|x| < \sqrt{\frac{\hbar}{2}}$ the integral becomes more intricate:
\begin{equation}
\begin{array}{c}
\int_{\mathbb{R}} \mu_{h_1} (x,p) dp= \frac{e^{-\frac{x^2}{\hbar}}}{\pi \hbar\left(4/\sqrt{e}-1 \right)} \left\{\int_{|p| < \sqrt{\hbar/2-x^2}}  \left(1-\frac{2x^2}{\hbar}-\frac{2p^2}{\hbar}  \right) e^{-\frac{p^2}{\hbar}}dp + \right. \\
\\
+ \left. \int_{|p| \geq \sqrt{\hbar/2-x^2}}  \left(\frac{2x^2}{\hbar}+\frac{2p^2}{\hbar}  -1 \right) e^{-\frac{p^2}{\hbar}}dp \right\}  =\\
\\
=\frac{e^{-\frac{x^2}{\hbar}}}{\pi \hbar\left(4/\sqrt{e}-1 \right)} \left\{\int_{\mathbb{R}}  \left(1-\frac{2x^2}{\hbar}-\frac{2p^2}{\hbar}  \right) e^{-\frac{p^2}{\hbar}}dp + \right. \\
\\
+ \left. 4 \int_{ \sqrt{\hbar/2-x^2}}^{\infty}  \left(\frac{2x^2}{\hbar}+\frac{2p^2}{\hbar}  -1 \right) e^{-\frac{p^2}{\hbar}}dp \right\}=\\
\\
=\frac{e^{-\frac{x^2}{\hbar}}}{\pi \hbar\left(4/\sqrt{e}-1 \right)} \left\{\frac{2x^2\sqrt{\pi \hbar}}{\hbar} + \frac{4 e^{\frac{x^2}{\hbar}}}{\sqrt{e}} \sqrt{\frac{\hbar}{2}-x^2}  + \frac{4x^2}{\sqrt{\hbar}} \Gamma \left(\frac{1}{2},\frac{1}{2}- \frac{x^2}{\hbar} \right)\right\}~,
\end{array}
\label{eqExample16}
\end{equation}
where $\Gamma (s,x)= \int_x^{\infty} t^{s-1} e^{-t} dt $ is the upper incomplete gamma function.

Due to the symmetry $x \leftrightarrow p$ in (\ref{eqExample11}) the same results apply to the momentum marginal distributions.

This shows that, while the Wigner function has the correct marginal distributions, the distribution $\mu_{\rho}$ does not in general. The exceptions are, of course, Wigner-positive functions which satisfy: $\mu_{\rho}=W \rho$.
\end{example}

\begin{example}\label{ExampleConcavity}
Let us consider a function $f\in \mathcal{S}_0 (\mathbb{R})$ such that its support is some interval $I_f$. We may consider, e.g. $f \in \mathcal{D} (I_f)$. It is well known that the support of $Wf$ is contained in the infinite strip $\Omega_f= I_f \times \mathbb{R}$.

Let us also consider a density matrix $\widehat{\eta}$ whose Wigner function $W \eta(z)$ has a support $\Omega_{\eta} \subset \mathbb{R}^2$, which we assume to be disjoint of $\Omega_f$:
\begin{equation}
\Omega_{\eta} \cap \Omega_f = \emptyset~.
\label{eqExampleConcavity1}
\end{equation}
We shall also assume that:
\begin{equation}
K:= \frac{\|Wf\|_1}{\|W  \eta\|_1} >1~.
\label{eqExampleConcavity2}
\end{equation}
Finally, we also consider the convex combination:
\begin{equation}
W \rho (z)= \frac{W f(z)+ W \eta(z)}{2}~.
\label{eqExampleConcavity3}
\end{equation}
A straightforward calculation shows that:
\begin{equation}
\|W \rho\|_1 = \frac{\|W f\|_1+ \|W \eta \|_1}{2}~.
\label{eqExampleConcavity4}
\end{equation}
As before, let $\mu_f(z)= \frac{|Wf(z)|}{\|Wf\|_1}$ and $\mu_{\eta}(z)= \frac{|W \eta(z)|}{\|W \eta \|_1}$. We thus have:
\begin{equation}
\begin{array}{c}
S\left[W \rho \right]= - \int_{\mathbb{R}^2} \frac{|W\rho(z)|}{\|W\rho\|_1} \ln \left(\frac{|W\rho(z)|}{\|W\rho\|_1}\right) dz=\\
\\
= - \int_{\Omega_f} \frac{|W f(z)|}{2\|W\rho\|_1} \ln \left(\frac{|W f(z)|}{2\|W\rho\|_1}\right) dz - \int_{\Omega_{\eta}} \frac{|W\eta(z)|}{2\|W\rho\|_1} \ln \left(\frac{|W\eta(z)|}{2\|W\rho\|_1}\right) dz=\\
\\
= - \frac{\|W f\|_1}{2\|W\rho\|_1}\int_{\Omega_f} \mu_f(z) \ln \left(\frac{\mu_f(z) \|W f\|_1}{2\|W\rho\|_1}\right) dz - \frac{\|W \eta\|_1}{2\|W\rho\|_1}\int_{\Omega_{\eta}} \mu_{\eta}(z) \ln \left(\frac{\mu_{\eta}(z)\|W\eta\|_1}{2\|W\rho\|_1}\right) dz=\\
\\
=\frac{\|W f\|_1S \left[Wf\right]+\|W \eta\|_1S \left[W\eta\right]}{2\|W\rho\|_1} -\frac{\|W f\|_1}{2\|W\rho\|_1}\ln \left(\frac{\|W f\|_1}{2\|W\rho\|_1}\right)-\frac{\|W \eta\|_1}{2\|W\rho\|_1}\ln \left(\frac{\|W \eta\|_1}{2\|W\rho\|_1}\right)~.
\end{array}
\label{eqExampleConcavity5}
\end{equation}
Next we consider the quantity:
\begin{equation}
\Sigma:=S \left[W \rho\right]-\frac{1}{2}S\left[W f\right]-\frac{1}{2} S\left[W \eta\right]~.
\label{eqExampleConcavity6}
\end{equation}
From(\ref{eqExampleConcavity4}-\ref{eqExampleConcavity6}) we obtain:
\begin{equation}
\Sigma=\Sigma_1+\Sigma_2~,
\label{eqExampleConcavity7}
\end{equation}
where
\begin{equation}
\Sigma_1= \frac{\left(\|Wf\|_1-\|W \eta\|_1 \right)}{4 \|W \rho\|_1} \left(S \left[W f\right]-S \left[W \eta\right] \right)~,
\label{eqExampleConcavity8}
\end{equation}
and
\begin{equation}
\Sigma_2= -\frac{\|W f\|_1}{2\|W\rho\|_1}\ln \left(\frac{\|W f\|_1}{2\|W\rho\|_1}\right)-\frac{\|W \eta\|_1}{2\|W\rho\|_1}\ln \left(\frac{\|W \eta\|_1}{2\|W\rho\|_1}\right)~.
\label{eqExampleConcavity9}
\end{equation}
Using the fact that $1 < \|W \eta\|_1 < \|W f\|_1$, we conclude that:
\begin{equation}
\begin{array}{c}
\Sigma_2= \ln \left(2 \|W \rho\|_1\right) - \frac{\|W f\|_1\ln \|W f\|_1 + \|W \eta\|_1\ln \|W \eta\|_1}{2\|W \rho\|_1} <\\
\\
< \ln \left(2 \|W f\|_1\right)  - \frac{\|W f\|_1\ln \|W \eta\|_1 + \|W \eta\|_1\ln \|W \eta\|_1}{2 \|W \rho\|_1}=\\
\\
= \ln \left(2 \|W f\|_1\right) -\ln \|W \eta\|_1=\ln (2K)~.
\end{array}
\label{eqExampleConcavity10}
\end{equation}
From (\ref{eqExampleConcavity7},\ref{eqExampleConcavity8},\ref{eqExampleConcavity10}) we have:
\begin{equation}
\Sigma < \frac{\left(\|Wf\|_1-\|W \eta\|_1 \right)}{4 \|W \rho\|_1} \left(S \left[W f\right]-S \left[W \eta\right] \right) + \ln (2K)~.
\label{eqExampleConcavity11}
\end{equation}
We conclude that we may obtain $\Sigma <0$, if we can find $Wf$, $W \eta$, such that:
\begin{equation}
\begin{array}{c}
\frac{\left(\|Wf\|_1-\|W \eta\|_1 \right)}{4 \|W \rho\|_1} \left(S \left[W f\right]-S \left[W \eta\right] \right) + \ln (2K) \leq 0\\
\\
\Leftrightarrow  S \left[W \eta\right]-S\left[W f\right] \geq \frac{2(K+1)}{K-1} \ln (2K)~.
\end{array}
\label{eqExampleConcavity12}
\end{equation}
We now show that this is indeed possible.

Suppose that $I_f=\left[-1,0 \right]$ and consider some function $g_0 \in \mathcal{S}_0(\mathbb{R})$ with support $I_{g_0}= \left[0,1\right]$. Its Wigner distribution $W g_0$ has support $I_{g_0} \times \mathbb{R}$. We shall also assume that $\|Wg_0\|_1 <\|Wf\|_1$. 

Next, we define:
\begin{equation} 
W \eta_n (z)= \frac{1}{n} \left(W g_0(z-z_1)+W g_0(z-z_2)+ \cdots+W g_0(z-z_n) \right)~,
\label{eqExampleConcavity13}
\end{equation} 
for some $n \geq 2$ and where $z_j=(j,0)$. Thus $W g_0(z-z_1)$ has support $\left[1,2\right] \times \mathbb{R}$, $W g_0(z-z_2)$ has support $\left[2,3\right] \times \mathbb{R}$, and so on.

After a simple calculation, we conclude that:
\begin{equation}
\|W \eta_n\|_1 = \|Wg_0\|_1~.
\label{eqExampleConcavity14}
\end{equation} 
Thus 
\begin{equation}
K=\frac{\|Wf\|_1}{\|W \eta_n\|_1} =\frac{\|Wf\|_1}{\|W g_0\|_1} >1
\label{eqExampleConcavity15}
\end{equation} 
is constant for all $n \geq 2$. Thus the right-hand side of the second inequality in (\ref{eqExampleConcavity12}) does not change with $n$. It remains to prove that, by increasing $n$, we can make $S\left[W \eta_n \right]$ as large as we wish:
\begin{equation}
\begin{array}{c}
S \left[W \eta_n \right]= - \frac{1}{n} \sum_{j=1}^n \int_{\left[j,j+1\right] 
\times \mathbb{R}} \frac{|Wg_0(z-z_j)|}{\|g_0\|_1} \ln \left(\frac{|Wg_0(z-z_j)|}{n\|g_0\|_1} \right) dz=\\
\\
=-\int_{\left[0,1\right] 
\times \mathbb{R}} \frac{|Wg_0(z)|}{\|g_0\|_1} \ln \left(\frac{|Wg_0(z)|}{n\|g_0\|_1} \right) dz =S \left[Wg_0 \right]+ \ln (n)~.
\end{array}
\label{eqExampleConcavity16}
\end{equation}

\end{example}

***************************************************************

Authors' address:

\vspace{0.5 cm}
\noindent
Grupo de F\'{\i}sica Matem\'{a}tica, Departamento de
Matem\'{a}tica, Faculdade de Ci\^{e}ncias, Universidade de Lisboa, Campo
Grande, Edif\'{\i}cio C6, 1749-016 Lisboa, Portugal

\end{document}